\documentclass[aps,prd,preprint,nofootinbib]{revtex4-2}

\usepackage[T1]{fontenc}
\usepackage{lmodern}
\usepackage{microtype}
\usepackage{amsmath,amssymb,amsthm}
\usepackage{bm}
\usepackage{graphicx}
\usepackage{xcolor}
\usepackage{hyperref}

\definecolor{linkblack}{RGB}{20,20,20}
\definecolor{citeblue}{RGB}{0,70,140}
\definecolor{urlblue}{RGB}{0,90,120}
\hypersetup{
  colorlinks=true,
  linkcolor=linkblack,
  citecolor=citeblue,
  urlcolor=urlblue
}
\newtheorem{theorem}{Theorem}
\newtheorem{proposition}{Proposition}

\begin{document}

\title{Obstructions to Minimal Regular Black Hole Cosmologies}

\author{Damien A. Easson}
\email{easson@asu.edu}
\affiliation{Department of Physics, Arizona State University, Tempe, AZ 85287, USA}
\affiliation{Beyond Center for Fundamental Concepts in Science, Arizona State University, Tempe, AZ 85287, USA}

\begin{abstract}
We derive an obstruction to Friedmann--Lema\^itre--Robertson--Walker (FLRW) daughter cosmologies from
static, asymptotically flat regular black holes. The trapped region of such a
parent is Kantowski--Sachs rather than FLRW, so the daughter must be introduced
as a separate matched region. For closed daughters, the angular Darmois
condition is controlled by the Misner--Sharp mass: asymptotic flatness and
finite ADM mass force the induced density to decay as $A^{-3}$, while the
$k=+1$ curvature term scales as $A^{-2}$. The minimal closed branch is
therefore bounded rather than indefinitely expanding. Flat and open daughters
avoid this boundedness mechanism. For their maximal homogeneous FLRW
continuations, the regular-end affine-ANEC theorem excludes simultaneous
non-staticity, curvature regularity, null geodesic completeness, and ANEC
consistency. For Bardeen, the parent source does not naturally supply the
late-time support needed for an unbounded closed daughter. Within these
criteria, a viable FLRW daughter therefore requires additional structure, such
as modified asymptotics, nonminimal matching, non-FLRW evolution, or an
additional stress-energy component.
\end{abstract}

\maketitle
\newpage

\section{Introduction}
\label{sec:intro}

A recurring idea in gravitational physics is that black-hole interiors may
admit a cosmological interpretation. In singular examples such as
Schwarzschild, the trapped region can be rewritten as a Kantowski--Sachs
spacetime, while in regular black holes the singularity is replaced by a
nonsingular core, often locally de Sitter-like near the center. This raises a
simple question: can a regular black-hole interior furnish a genuine daughter
cosmology, or does such an interpretation require additional global structure
or a new stress-energy component?

To address this question, we consider static, spherically symmetric,
asymptotically flat regular black holes with trapped interiors, using the
Bardeen black hole as a representative example. By a \emph{minimal daughter construction} we mean a no-shell FLRW matching
whose evolution is fixed by the parent metric, with no added late-time bulk
sector, modified asymptotics, or independent shell stress tensor. The FLRW
daughter is therefore a separate homogeneous and isotropic matched
continuation rather than a coordinate rewriting of the Kantowski--Sachs
trapped region. Here ``no shell'' means no Israel surface stress tensor; we do
not assume that the parent and daughter arise from one globally defined
microscopic matter field. Completeness statements about flat/open daughters
refer to the maximal homogeneous FLRW continuation determined by the matched
scale-factor law, not to an otherwise unspecified global parent-plus-daughter
composite. We impose the averaged null energy condition (ANEC) as a minimal
non-exoticity criterion: it allows controlled local violations of the
pointwise null energy condition while constraining the averaged null energy
along complete null geodesics \cite{Wall:2009wi,Curiel:2014zba}. The
completeness theorem used below is formulated under the same assumption and
for the regular endpoint classes specified there
\cite{Burwig:2025hrr,Burwig:2026fsy}.

Throughout, an indefinitely expanding daughter means an unbounded branch with
scale factor $A(\tau)\to\infty$, rather than a long-lived finite expansion between turning
points. The main flat/open/closed dichotomy is formulated for the one-function
static spherical class $g_{tt}g_{rr}=-1$; the closed-branch boundedness is
extended to redshifted static parents in Sec.~\ref{subsec:redshift_functions}.\footnote{We take $c=G=1$, except when $G$ is restored for clarity in the
effective FLRW density and pressure.}

Historically, the broad idea that black-hole interiors may be related to
cosmological regions appears in black-hole-universe proposals
\cite{Pathria:1972,Easson:2001qf}, limiting-curvature and
baby-universe constructions
\cite{Blau:1986cw,Frolov:1989,Frolov:1990,FarhiGuthGuven:1990,Trodden:1993xh,Easson:2017pfe,Frolov:2021kcv},
and regular de Sitter-core scenarios
\cite{Bardeen:1968,Hayward:2006,Dymnikova:2001,AyonBeatoGarcia:2000,Dymnikova:2019}. A complementary line of work concerns
the geodesic completeness of nonsingular cosmologies
\cite{Borde:1994,BordeGuthVilenkin:2003,LesnefskyEassonDavies:2023,Easson:2024fzn,Easson:2024uxe,Easson:2026ret}.
More recently, Oppenheimer--Snyder-type and generalized constructions have
been developed for regular black holes and nonsingular cosmologies with
Hayward, Bardeen, and Minkowski-type cores~\cite{Shojai:2022,Bobula:2024,Bueno:2025,LiWuGe:2025,Bobula:2026zlq}.

We show that a regular core is not by itself enough to produce an FLRW daughter
having the full list of desired properties: indefinite expansion, curvature
regularity, geodesic completeness, and ANEC consistency. The result separates
into two logically distinct obstructions. For closed daughters, boundedness is
controlled by asymptotic flatness and finite ADM mass through the no-shell
matching equation. For flat and open daughters, the obstruction is more
general: for maximal FLRW continuations in the regular endpoint classes of the
cited theorem, it follows from the affine-ANEC completeness result independently
of the black-hole matching construction.

The rest of the paper is organized as follows. In Sec.~\ref{sec:setup_parent}
we review the trapped-region geometry of the parent spacetime. In
Sec.~\ref{sec:static_obstruction} we show that the trapped interior is not
itself an FLRW cosmology. In Sec.~\ref{sec:daughter_matching} we derive the
no-shell FLRW daughter equations and discuss a redshift-function extension.
In Sec.~\ref{sec:dichotomy} we prove the closed-branch boundedness theorem,
explain how the flat/open branches fall under the general FLRW completeness
obstruction, and analyze the expansion duration and tuning of the closed Bardeen branch. In
Sec.~\ref{sec:discussion} we discuss core type and source limitations.

\section{Trapped interiors of regular black holes}
\label{sec:setup_parent}

We begin by considering static,
spherically symmetric metrics of the form
\begin{equation}
ds^2=-f(r)\,dt^2+f(r)^{-1}dr^2+r^2d\Omega_2^2,
\label{eq:parent_static_metric}
\end{equation}
with an asymptotically flat exterior and a regular core
\cite{Bardeen:1968,AyonBeatoGarcia:2000,Hayward:2006,Ansoldi:2008,BalartVagenas:2014,Fan:2016rih}. On any trapped
interval \(f(r)<0\), defining
\begin{equation}
T\equiv r,\qquad \chi\equiv t,\qquad
d\tau=\frac{dT}{\sqrt{|f(T)|}},
\label{eq:trapped_coords_compact}
\end{equation}
recasts the metric as
\begin{align}
ds^2&=-d\tau^2+a(\tau)^2d\chi^2+b(\tau)^2d\Omega_2^2,
\nonumber\\
a(\tau)&=\sqrt{|f(T(\tau))|},
\qquad
b(\tau)=T(\tau).
\label{eq:ks_metric_compact}
\end{align}
For this one-function class, one has \(\dot b=a\) with the orientation chosen
so that \(T\) increases with \(\tau\); equivalently, \(|\dot b|=a\) independent
of orientation.

We choose to work with the Bardeen geometry
\cite{Bardeen:1968,AyonBeatoGarcia:2000},
\begin{equation}
f(r)=1-\frac{2Mr^2}{(r^2+g^2)^{3/2}}.
\label{eq:bardeen_f_compact}
\end{equation}
For \(M/g>3\sqrt{3}/4\), this solution has two horizons. Its causal structure
then consists of a static exterior, a finite trapped nonstatic band
\(r_-<r<r_+\), and a regular static core \(0\le r<r_-\). Near the center,
\begin{equation}
f(r)=1-\frac{2M}{g^3}r^2+O(r^4),
\qquad
\Lambda_{\rm eff}=\frac{6M}{g^3},
\label{eq:bardeen_core_compact}
\end{equation}
so the core is locally de Sitter-like. The effective source may be written as
\begin{equation}
T^\mu{}_{\nu}=\mathrm{diag}\bigl(-\rho(r),p_r(r),p_T(r),p_T(r)\bigr),
\label{eq:bardeen_source_tensor_compact}
\end{equation}
with
\begin{align}
\rho(r)&=\frac{3Mg^2}{4\pi(r^2+g^2)^{5/2}},
\nonumber\\
p_r(r)&=-\rho(r),
\nonumber\\
p_T(r)&=\frac{3Mg^2(3r^2-2g^2)}{8\pi(r^2+g^2)^{7/2}}.
\label{eq:bardeen_source_compact}
\end{align}
Thus the source is anisotropic away from the center and isotropizes only in the
\(r\to0\) limit.~\footnote{The same trapped-region geometry can also be used in
time-dependent classical double-copy constructions
\cite{Easson:2020esh,Easson:2026wex}. Here we use only the geometric and
source-side ingredients needed for the cosmological obstruction.}

\section{The trapped interior is not FLRW}
\label{sec:static_obstruction}

The following familiar observation sets our stage: the trapped
region of a static spherical black hole is Kantowski--Sachs rather than FLRW.
Thus the FLRW daughter considered below is a separate matched region, and not simply a
coordinate reinterpretation of the trapped band.

\begin{proposition}[Static trapped-region obstruction]
\label{prop:static_obstruction}
Consider a one-function static, spherically symmetric parent spacetime of the
form \eqref{eq:parent_static_metric} with a trapped interval $f(r)<0$. In the
induced Kantowski--Sachs description \eqref{eq:ks_metric_compact}, the trapped
interior is not an exact FLRW cosmology in its natural slicing.

Moreover, defining
\begin{equation}
H_\parallel:=\frac{\dot a}{a},
\qquad
H_\perp:=\frac{\dot b}{b},
\end{equation}
isotropic expansion on a connected open interval, $H_\parallel=H_\perp$, is
equivalent to
\begin{equation}
a(\tau)=C\,b(\tau),
\end{equation}
for some constant $C>0$, and hence to
\begin{equation}
|f(T)|=C^2T^2.
\end{equation}
Thus isotropic expansion can occur only in this highly special case and is
nongeneric within the one-function static class. This is a necessary condition
for isotropic expansion, not a sufficient condition for an exact FLRW geometry.
\end{proposition}

\begin{proof}
At fixed $\tau$, the induced spatial metric is
\begin{equation}
h_{ij}dx^i dx^j
=
a(\tau)^2 d\chi^2+b(\tau)^2 d\Omega_2^2 .
\end{equation}
Thus each spatial slice has the direct-product geometry
$\mathbb R_\chi\times S^2_{\,b(\tau)}$, lacking constant
sectional curvature: two-planes containing the $\chi$-direction have zero
sectional curvature, while two-planes tangent to the sphere have sectional
curvature $1/b(\tau)^2$. Therefore the trapped interior is not an exact FLRW
cosmology in its natural slicing.

In the one-function case, one has $\dot b=a$. A necessary condition for exact
FLRW behavior on an open interval is isotropic expansion,
$H_\parallel=H_\perp$, or equivalently
\begin{equation}
\frac{\dot a}{a}=\frac{\dot b}{b}.
\end{equation}
On a connected open interval this integrates to $a(\tau)=C\,b(\tau)$ for some
constant $C>0$, and the converse is immediate. Since $b=T$ and
$a=\sqrt{|f(T)|}$, the condition becomes
\begin{equation}
|f(T)|=C^2T^2.
\end{equation}

This is a highly special condition, and it still does not make the geometry
FLRW. Even when it holds, it only equalizes the directional Hubble rates; it
does not alter the intrinsic product geometry of the $\tau=\mathrm{const}$
slices. The natural spatial slices remain $\mathbb R\times S^2$ rather than
constant-curvature three-spaces.
\end{proof}

For the Bardeen metric \eqref{eq:bardeen_f_compact}, one has on the trapped
interval
\begin{equation}
|f(T)|=\frac{2MT^2}{(T^2+g^2)^{3/2}}-1,
\label{eq:bardeen_absf_obstruction}
\end{equation}
which is not of the form \(C^2T^2\) on any open interval. Hence the Bardeen
trapped region does not admit isotropic expansion over any finite portion of
its interior evolution. More importantly, the trapped region is only the finite
band \(r_-<r<r_+\): it terminates at the inner horizon and is followed by a
regular static core. Thus even though the core becomes locally de Sitter-like
as \(r\to0\), the trapped Kantowski--Sachs phase itself does not evolve into a
daughter FLRW branch.

Hence, Proposition~\ref{prop:static_obstruction} rules out the simplest possibility:
the trapped region of the static parent is not already the daughter universe.
For the Bardeen trapped band this conclusion is corroborated by the invariants.
FLRW spacetimes are conformally flat, whereas, for a one-function spherical
metric written as \(F(r)=1-2m(r)/r\), the quadratic Weyl invariant is
\begin{align}
C_{abcd}C^{abcd}
&=
\frac{\bigl[r^2F''-2rF'+2F-2\bigr]^2}{3r^4}
\nonumber\\
&=
\frac{48}{r^6}
\left(m-\frac{2r m'}{3}+\frac{r^2m''}{6}\right)^2 .
\label{eq:weyl_mass_function}
\end{align}
For Bardeen, \(m(r)=Mr^3/(r^2+g^2)^{3/2}\), so
\begin{equation}
C_{abcd}C^{abcd}
=
\frac{12M^2r^4(2r^2-3g^2)^2}{(r^2+g^2)^7}.
\label{eq:bardeen_weyl_invariant}
\end{equation}
This invariant vanishes at \(r=0\) and at the isolated sphere
\(r=\sqrt{3/2}\,g\), but it is not identically zero on any open radial
interval. Hence no open portion of the Bardeen trapped band is conformally
flat, and the trapped band is not an FLRW spacetime.

\section{No-shell FLRW daughters from static parents}
\label{sec:daughter_matching}

Having shown the Kantowski--Sachs trapped interior of the parent is not itself the daughter cosmology, we now turn to the simplest global continuation one might
try: a no-shell matching to an FLRW daughter region. In this matching,
the FLRW daughter is a separate spacetime region attached across a spherical
hypersurface and constrained by the Darmois--Israel no-shell conditions
\cite{OppenheimerSnyder:1939,Israel:1966}. This is the standard spherical
FLRW/static matching structure we use as a minimal control problem rather
than as a full collapse model.

The use of an FLRW daughter is not arbitrary, since for regular
black-holes such as Bardeen, the source isotropizes and the geometry
approaches a de Sitter-like core near \(r=0\). Since de Sitter space admits
homogeneous and isotropic slicings, an FLRW daughter is the natural optimistic
continuation to test. The question now is whether the most
symmetric no-shell daughter ansatz suggested by the regular core can support
an indefinitely expanding cosmology without additional global structure or some
new matter content.

\subsection{General junction equation}
\label{subsec:general_junction}

Here the \(+\) label denotes the static parent region, while the \(-\) label
denotes the FLRW daughter region.
We consider a static parent geometry, written in curvature coordinates as
\begin{equation}
ds_+^2=-F(R)\,dT^2+\frac{dR^2}{F(R)}+R^2 d\Omega_2^2,
\label{eq:generic_exterior_metric}
\end{equation}
and an FLRW daughter region
\begin{equation}
ds_-^2=-d\tau^2+A(\tau)^2\left[\frac{d\chi^2}{1-k\chi^2}+\chi^2 d\Omega_2^2\right],
\qquad
k=0,\pm1.
\label{eq:generic_frw_metric}
\end{equation}
We define
\begin{equation}
H:=\frac{\dot A}{A}.
\label{eq:frw_hubble}
\end{equation}

For the flat or open cases, the matching hypersurface is at fixed comoving
radius
\begin{equation}
\chi=\chi_b=\text{const},
\qquad
R_b(\tau)=A(\tau)\chi_b.
\label{eq:boundary_flat_open}
\end{equation}
On the FLRW side, the induced metric on the hypersurface is
\begin{equation}
ds^2_\Sigma=-d\tau^2+R_b(\tau)^2d\Omega_2^2.
\label{eq:induced_metric_interior}
\end{equation}
On the exterior side, parameterizing the hypersurface by
\(x_+^\mu(\tau,\theta,\varphi)=(T(\tau),R_b(\tau),\theta,\varphi)\), one finds
\begin{equation}
ds^2_\Sigma=-\left(F(R_b)\dot T^{\,2}-\frac{\dot R_b^{\,2}}{F(R_b)}\right)d\tau^2
+R_b^2d\Omega_2^2,
\label{eq:induced_metric_exterior}
\end{equation}
so matching the first fundamental form gives
\begin{equation}
F(R_b)\dot T^{\,2}-\frac{\dot R_b^{\,2}}{F(R_b)}=1.
\label{eq:tt_matching_condition}
\end{equation}
With the standard orientation choice for the outward normal, the angular
extrinsic-curvature condition gives
\begin{equation}
\dot R_b^{\,2}+F(R_b)=1-k\chi_b^2.
\label{eq:rr_matching_condition}
\end{equation}
The remaining \(\tau\tau\) component is consistent with the same comoving
no-shell construction; for completeness we derive both components in
Appendix~\ref{app:darmois}.
Using \(R_b=A\chi_b\) and dividing by \(A^2\chi_b^2\), one obtains
\begin{equation}
H^2+\frac{k}{A^2}
=
\frac{1-F(A\chi_b)}{A^2\chi_b^2}.
\label{eq:general_matching_equation}
\end{equation}
Equation~\eqref{eq:general_matching_equation} is the key point of the
minimal construction: with no shell stress tensor, the daughter Friedmann
equation is inherited directly from the parent metric function $F(R)$.
Equivalently, defining the Misner--Sharp mass function of the static parent geometry
\cite{MisnerSharp:1964} by
\begin{equation}
 m(R):=\frac{R}{2}\bigl[1-F(R)\bigr],
\label{eq:misner_sharp_mass}
\end{equation}
Eq.~\eqref{eq:general_matching_equation} becomes the invariant relation
\begin{equation}
H^2+\frac{k}{A^2}
=
\frac{2m(A\chi_b)}{A^3\chi_b^3}
\qquad (k=0,-1),
\label{eq:general_matching_mass_function}
\end{equation}
and, for the closed parametrization of
Sec.~\ref{subsec:closed_daughters}, in which \(\psi_b\) is the fixed comoving
angular position of the closed FLRW boundary and \(R_b=A\sin\psi_b\),
\begin{equation}
H^2+\frac{1}{A^2}
=
\frac{2m(A\sin\psi_b)}{A^3\sin^3\psi_b}.
\label{eq:closed_matching_mass_function}
\end{equation}
This form makes clear that the matching equation is controlled by the parent
mass profile rather than by a coordinate artifact. Although the preceding
derivation was written in static coordinates away from coordinate horizons,
the resulting no-shell evolution equation extends by continuity across simple
zeros of \(F\).

\subsection{Redshift functions do not rescue the closed branch}
\label{subsec:redshift_functions}

A natural question is whether the closed-branch obstruction is an artifact of the
special choice \(g_{TT}g_{RR}=-1\). To test this, we allow a nontrivial
static redshift function. Under the same comoving no-shell assumptions, the
angular Darmois condition remains controlled by the Misner--Sharp mass
profile, while the redshift function enters only through the
\(K_{\tau\tau}\) condition.

Consider the more general static spherical
metric
\begin{equation}
 ds_+^2=-e^{2\Phi(R)}F(R)dT^2+\frac{dR^2}{F(R)}+R^2d\Omega_2^2 \,,
\label{eq:general_redshift_metric}
\end{equation}
where \(\Phi(R)\) is the redshift function.
For a comoving FLRW boundary, the angular Darmois condition is unchanged:
\begin{equation}
 \dot R_b^{\,2}+F(R_b)=\beta^2,
\label{eq:redshift_angular_condition}
\end{equation}
where
\begin{equation}
\beta^2=
\begin{cases}
1-k\chi_b^2, & k=0,-1,\\
\cos^2\psi_b, & k=+1.
\end{cases}
\label{eq:beta_definition}
\end{equation}
Thus, for closed daughters, the Friedmann equation and its large-\(A\)
asymptotics are controlled by \(F(R)\), equivalently by the Misner--Sharp mass,
and not by the redshift function \(\Phi\). If
\(F(R)=1-2M/R+o(R^{-1})\), the same asymptotic obstruction to unbounded closed
expansion follows.

The $\tau\tau$ condition is more restrictive: because the FLRW boundary is
comoving, $K^-_{\tau\tau}=0$, and hence the parent-side trajectory must be
geodesic. For Eq.~\eqref{eq:general_redshift_metric}, the radial geodesic
equation is
\begin{equation}
\ddot R_b=-\frac12 F'(R_b)
-\Phi'(R_b)\bigl(\dot R_b^{\,2}+F(R_b)\bigr).
\label{eq:redshift_geodesic_equation}
\end{equation}
On any nonstatic trajectory, differentiating
Eq.~\eqref{eq:redshift_angular_condition} away from isolated turning points
gives
\begin{equation}
\ddot R_b=-\frac12F'(R_b),
\label{eq:redshift_differentiated_angular}
\end{equation}
and the equality extends through such turning points by continuity. Combining
Eqs.~\eqref{eq:redshift_geodesic_equation} and
\eqref{eq:redshift_differentiated_angular} gives
\begin{equation}
\Phi'(R_b)\bigl(\dot R_b^{\,2}+F(R_b)\bigr)=0.
\end{equation}
Using Eq.~\eqref{eq:redshift_angular_condition}, this becomes
\begin{equation}
\Phi'(R_b)\,\beta^2=0.
\label{eq:redshift_no_shell_constraint}
\end{equation}
For flat/open boundaries and for non-equatorial closed boundaries,
$\beta^2>0$, so a nonstatic full Darmois matching requires
$\Phi'(R_b)=0$ throughout the swept trajectory. In the equatorial closed case
$\beta^2=0$, Eq.~\eqref{eq:redshift_no_shell_constraint} imposes no condition
on $\Phi$, but the angular condition itself gives
$\dot R_b^{\,2}+F(R_b)=0$, which forbids unbounded expansion when $F\to1$.

An identically static boundary is not covered by differentiating the first
integral. If $R_b=R_0$ is timelike and static, the angular condition gives
$F(R_0)=\beta^2$, while the remaining Darmois condition must be imposed
separately:
\[
\frac12F'(R_0)+\Phi'(R_0)F(R_0)=0.
\]
For $\Phi=0$ this reduces to $F'(R_0)=0$. Thus a redshift function does not
alter the large-radius closed-branch obstruction, which remains controlled by
the invariant mass profile.

\subsection{Closed daughter universes}
\label{subsec:closed_daughters}

For the closed case \(k=+1\), it is convenient to write the daughter metric as
\begin{equation}
ds_-^2=-d\tau^2+A(\tau)^2\left(d\psi^2+\sin^2\psi\,d\Omega_2^2\right),
\label{eq:closed_frw_metric}
\end{equation}
with the boundary at fixed
\begin{equation}
\psi=\psi_b=\text{const},
\qquad
R_b(\tau)=A(\tau)\sin\psi_b.
\label{eq:boundary_closed}
\end{equation}
Since the areal radius depends on \(\sin\psi_b\), we take
\(0<\psi_b\le\pi/2\) without loss of generality for the nondegenerate
branch considered below.
The matching equation becomes
\begin{equation}
H^2+\frac{1}{A^2}
=
\frac{1-F(A\sin\psi_b)}{A^2\sin^2\psi_b}.
\label{eq:closed_matching_equation}
\end{equation}

For the Bardeen exterior, this gives
\begin{equation}
H^2
=
\frac{2M}{\bigl(A^2\sin^2\psi_b+g^2\bigr)^{3/2}}
-\frac{1}{A^2},
\label{eq:bardeen_closed_matching}
\end{equation}
which clearly displays the competition between the finite-mass contribution and the
$k=+1$ positive-curvature term. At large $A$, the positive contribution falls
as $A^{-3}$, whereas the curvature term falls as $A^{-2}$ and therefore
dominates. In the next section we show that this obstruction is not special to
Bardeen, but follows from asymptotic flatness and finite ADM mass. As shown in
Sec.~\ref{subsec:redshift_functions}, adding a static redshift function does
not change the large-radius conclusion.

We may express the same result as an effective closed-FLRW source. Below, the
density $\rho_{\rm eff}$ is the homogeneous density inferred from the daughter
Friedmann equation and not the unchanged anisotropic Bardeen
nonlinear-electrodynamics stress tensor. Writing
\begin{equation}
H^2+\frac{1}{A^2}
=
\frac{8\pi G}{3}\rho_{\rm eff}(A),
\label{eq:effective_friedmann_closed_matching}
\end{equation}
one obtains
\begin{equation}
\rho_{\rm eff}(A)
=
\frac{3M}{4\pi G\,\bigl(A^2\sin^2\psi_b+g^2\bigr)^{3/2}}.
\label{eq:rhoeff_bardeen_matching}
\end{equation}
As a function of the boundary areal radius $R_b=A\sin\psi_b$, this
effective source approaches vacuum-like behavior only in the deep-core limit
$R_b\ll g$, and it falls as $A^{-3}$ at large $A$.

The corresponding effective pressure is fixed by the
conservation equation,
\[
\frac{d\rho_{\rm eff}}{d\tau}
+3H(\rho_{\rm eff}+p_{\rm eff})=0,
\]
or, equivalently,
\[
p_{\rm eff}(A)
=
-\rho_{\rm eff}(A)-\frac{A}{3}\frac{d\rho_{\rm eff}}{dA}.
\]
For Bardeen this gives (with
\(s=\sin\psi_b\) for the closed branch),
\[
p_{\rm eff}(A)
=
-\frac{3Mg^2}
{4\pi G\bigl(A^2s^2+g^2\bigr)^{5/2}},
\]
and hence
\[
w_{\rm eff}(A)
=
\frac{p_{\rm eff}}{\rho_{\rm eff}}
=
-\frac{g^2}{A^2s^2+g^2}.
\]
Thus, as a function of $R_b=A s$, the matched source approaches
vacuum-like behavior $w_{\rm eff}\to-1$ for $R_b/g\to0$ and dust-like
behavior $w_{\rm eff}\to0$ at large $R_b/g$. The physical inner turning
point occurs at finite $R_b$, so the equation of state at the bounce need not
be parametrically close to $-1$. The source supplies acceleration only for
\(A^2s^2<2g^2\), and therefore cannot provide the persistent
\(w\le -1/3\) support needed for an unbounded closed daughter.

\subsection{Flat daughter universes}
\label{subsec:flat_daughters}

For the flat case \(k=0\), Eq.~\eqref{eq:general_matching_equation} reduces to
\begin{equation}
H^2
=
\frac{1-F(A\chi_b)}{A^2\chi_b^2}.
\label{eq:flat_matching_equation}
\end{equation}
For Bardeen,
\begin{equation}
H^2
=
\frac{2M}{\bigl(A^2\chi_b^2+g^2\bigr)^{3/2}}.
\label{eq:bardeen_flat_matching}
\end{equation}
Unlike the closed case, there is no positive-curvature term forcing a turning
point, so the flat daughter branch is monotonic. Near \(A\to0\),
\begin{equation}
H^2\to \frac{2M}{g^3},
\label{eq:bardeen_flat_early_matching}
\end{equation}
and the early-time behavior is asymptotically de Sitter-like,
\begin{equation}
A(\tau)\sim e^{H_0\tau},
\qquad
H_0^2=\frac{2M}{g^3}.
\label{eq:bardeen_flat_desitter_matching}
\end{equation}
At large \(A\),
\begin{equation}
H^2\sim \frac{2M}{\chi_b^3A^3},
\label{eq:bardeen_flat_late_matching}
\end{equation}
so
\begin{equation}
A(\tau)\propto \tau^{2/3}.
\label{eq:bardeen_flat_late_solution_matching}
\end{equation}
Thus the flat daughter branch is nonsingular in cosmic time and evolves from
an asymptotically de Sitter-like past into a dust-like late phase. As we will
emphasize below, however, this does not furnish a geodesically complete
nonsingular cosmology.

\subsection{Open daughter universes}
\label{subsec:open_daughters}

For the open case \(k=-1\), Eq.~\eqref{eq:general_matching_equation} gives
\begin{equation}
H^2-\frac{1}{A^2}
=
\frac{1-F(A\chi_b)}{A^2\chi_b^2}.
\label{eq:open_matching_equation}
\end{equation}
For an asymptotically flat finite-mass parent,
\begin{equation}
H^2=
\frac{1}{A^2}
+
\frac{2M}{A^3\chi_b^3}
+o(A^{-3}).
\label{eq:open_late_H2}
\end{equation}
Equation~\eqref{eq:open_late_H2} is the large-\(A\) asymptotic limit of
Eq.~\eqref{eq:open_matching_equation}.
Thus the open branch, like the flat branch, avoids the closed-universe
recollapse argument. Its obstruction is instead the shared flat/open
completeness theorem used in Sec.~\ref{sec:dichotomy}.

If the parent has a de Sitter-like core,
\(F(R)=1-H_0^2R^2+O(R^4)\), then the open matching equation gives
\[
H^2-\frac{1}{A^2}\to H_0^2,
\qquad
\dot A^2=1+H_0^2A^2+O(A^4).
\]
Locally, the leading-order solution is the open de Sitter slicing,
\[
A(\tau)=H_0^{-1}\sinh\!\bigl[H_0(\tau-\tau_0)\bigr]+O(A^3).
\]
The endpoint \(A=0\) is not a curvature singularity, but the open slicing
covers only a geodesically incomplete patch of de Sitter space. Indeed,
$A(\tau)\sim\tau-\tau_0$ there, so a radial null geodesic reaches the endpoint
in finite affine parameter:
\[
\Delta\lambda\propto
\int_{\tau_0}^{\tau_1}A(\tau)\,d\tau<\infty.
\]
This provides the local model for the flat/open incompleteness obstruction
discussed in Sec.~\ref{sec:dichotomy}.

\section{Obstruction mechanisms}
\label{sec:dichotomy}

The preceding matching equations give the closed-branch obstruction, while the
flat/open branches are controlled by the general FLRW completeness theorem.

\begin{theorem}[Minimal asymptotically flat FLRW daughters]
\label{thm:minimal_frw_daughters}
Consider a static, spherically symmetric, asymptotically flat parent spacetime
in the one-function class $g_{tt}g_{rr}=-1$, with finite ADM mass $M>0$.
Attach an FLRW daughter across a nondegenerate comoving spherical Darmois
boundary. Assume that the matching is no-shell and that the daughter evolution
is fixed by the parent metric profile, with no additional late-time bulk
component, modified asymptotics, or independent shell stress tensor. Then the
closed daughter branch is bounded at finite scale factor and therefore does
not yield an indefinitely expanding FLRW cosmology.

For $k=0$ and $k=-1$, consider the maximal homogeneous FLRW continuation
determined by the matched scale-factor law. If that continuation is nonstatic,
curvature-regular, and has regular affine ends in the sense of
Refs.~\cite{Burwig:2025hrr,Burwig:2026fsy}, it cannot be both null
geodesically complete and ANEC-consistent. Consequently, within these explicit
hypotheses, the minimal construction does not furnish an FLRW continuation
satisfying all four desired properties: indefinite expansion, curvature
regularity, geodesic completeness, and ANEC consistency.
\end{theorem}

The proof separates into two logically distinct ingredients. The closed branch
is bounded by the finite-ADM matching equation. The cited flat/open theorem is
a conditional statement about the regular endpoint classes specified in those
references; since it obstructs null completeness, it also obstructs full
geodesic completeness within that class. For the explicit Bardeen flat and
open branches, null incompleteness is established directly by the affine-length
estimates in Eq.~\eqref{eq:flat_affine_incompleteness} and
Sec.~\ref{subsec:open_daughters}, respectively. These statements concern maximal FLRW
continuations and do not, by themselves, decide the geodesic completeness of a
separately specified global parent-plus-daughter composite. As shown in
Sec.~\ref{subsec:redshift_functions}, the closed-branch boundedness also
extends to redshifted static parents in the same comoving no-shell setup.

\subsection{Closed daughters from asymptotically flat parents are bounded}
\label{subsec:closed_recollapse}

We now assume only that the parent geometry is asymptotically flat with finite
ADM mass,
\begin{equation}
F(R)=1-\frac{2M}{R}+o(R^{-1}),
\qquad M>0.
\label{eq:asymptotic_flatness}
\end{equation}
The resulting boundedness statement is given in
Proposition~\ref{prop:closed_recollapse}, and the mechanism is illustrated in
Fig.~\ref{fig:bardeen_recollapse_obstruction} for a representative
asymptotically flat Bardeen daughter. The important point is that the
obstruction is not the absence of a bounce or the formation of a singular big
crunch. The branch is bounded between finite turning points. The obstruction is
instead that the minimal closed branch does not yield an indefinitely expanding
daughter universe.

\begin{figure}[t]
\centering
\includegraphics[width=\columnwidth]{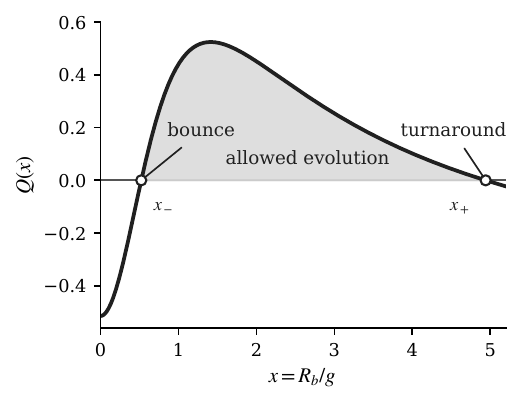}
\caption{Closed Bardeen daughter evolution in the asymptotically flat case.
The plotted function is
\(Q(x):=2\mu x^2/(1+x^2)^{3/2}-\sin^2\psi_b\), with
\(x=R_b/g=A\sin\psi_b/g\), \(\eta=\tau/g\), \(\mu=M/g\), and representative
parameters \((\mu,\psi_b)=(1.35,0.80)\). On the dynamically allowed interval,
\(Q(x)=(dx/d\eta)^2\ge0\), which is the shaded region. The two zeros \(x_-\) and \(x_+\) are
finite turning points: a nondegenerate expanding branch can emerge from the
inner turning point and reach an outer turnaround before recollapsing. The obstruction is the boundedness of the closed branch in the asymptotically
flat minimal construction, with evolution between finite turning points rather
than a singular big crunch.}
\label{fig:bardeen_recollapse_obstruction}
\end{figure}

\begin{proposition}[Closed daughter boundedness]
\label{prop:closed_recollapse}
For any no-shell closed FLRW daughter universe with nondegenerate matching
surface $0<\psi_b\le \pi/2$, matched to a static asymptotically flat parent
geometry with finite ADM mass and curvature-coordinate function
$F(R)=1-2M/R+o(R^{-1})$, the daughter branch cannot extend to arbitrarily
large $A$. If the allowed interval has a simple finite outer zero, the expanding branch
reaches a finite turning point and recollapses. Degenerate endpoints are tuned
limiting cases corresponding to static or asymptotically static behavior.
\end{proposition}
\begin{proof}
The closed mechanical form gives
\[
\dot R_b^{\,2}=E-F(R_b),
\qquad
E=\cos^2\psi_b .
\]
For $0<\psi_b\le \pi/2$, one has $E<1$, while asymptotic flatness gives
$F(R)\to1$. Hence
\[
E-F(R)\longrightarrow -\sin^2\psi_b<0 ,
\]
so the physical region $\dot R_b^{\,2}\ge0$ cannot extend to arbitrarily large
$R_b$. Since $R_b=A\sin\psi_b$ and $\sin\psi_b>0$, the daughter branch cannot
extend to arbitrarily large $A$.

Equivalently, the large-\(A\) Friedmann form is
\[
H^2=
-\frac{1}{A^2}
+\frac{2M}{A^3\sin^3\psi_b}
+o(A^{-3})
\qquad (A\to\infty).
\]
Thus the $A^{-2}$ curvature term eventually dominates the induced $A^{-3}$
matter term, and $H^2<0$ at sufficiently large scale factor. Along a physical
daughter solution one must have $H^2\ge0$. Therefore an expanding branch cannot
be unbounded.

If the upper endpoint of an allowed connected component is a simple zero of
$E-F(R_b)$, the expanding solution reaches a finite outer turning point
$A_{\max}$ with $H(A_{\max})=0$ and then recollapses. Degenerate zeros are tuned
limiting cases and give static or asymptotically static behavior rather than
an indefinitely expanding daughter universe.
\end{proof}

This is a class statement depending only on the angular junction equation
and the asymptotic Schwarzschild falloff of the parent geometry, not on whether
the regular core is of de Sitter type, Minkowski type, or some other
nonsingular form. In mass-function language, the obstruction is simply the
finite-ADM limit \(m(R)\to M\).

\subsection{Expansion duration and tuning of the closed branch}
\label{subsec:duration_tuning}

The closed Bardeen branch illustrates that boundedness means a finite duration
for each expansion phase, not a finite spacetime lifetime or a singular
collapse.
Define
\begin{equation}
x:=\frac{R_b}{g}=\frac{A\sin\psi_b}{g},
\qquad
\eta:=\frac{\tau}{g},
\qquad
\mu:=\frac{M}{g}.
\label{eq:dimensionless_closed_vars}
\end{equation}
Equation~\eqref{eq:bardeen_closed_matching} becomes
\begin{equation}
\left(\frac{dx}{d\eta}\right)^2
=
\frac{2\mu x^2}{(1+x^2)^{3/2}}-\sin^2\psi_b.
\label{eq:bardeen_dimensionless_xdot}
\end{equation}
The first term on the right-hand side has maximum at \(x=\sqrt2\), with
value \(4\mu/(3\sqrt3)\). Thus an allowed nondegenerate closed Bardeen branch
exists when
\begin{equation}
\sin^2\psi_b<\frac{4\mu}{3\sqrt3},
\label{eq:bardeen_allowed_condition}
\end{equation}
while equality gives a tuned static limiting case. For a nonextremal Bardeen
black hole, \(\mu>3\sqrt3/4\), this condition is automatically satisfied for
\(0<\psi_b\le\pi/2\). The allowed region is therefore the finite interval
between the two positive roots \(x_-\) and \(x_+\) of the right-hand side.
The simple roots are regular turning points of the second-order equation
\[
\frac{d^2x}{d\eta^2}
=
\frac12\frac{d}{dx}
\left[
\frac{2\mu x^2}{(1+x^2)^{3/2}}-\sin^2\psi_b
\right].
\]
Consequently, the maximal FLRW scale-factor solution can be continued smoothly
through the outer turnaround and the inner bounce, giving a nonsingular
periodic evolution. It fails the criterion adopted here because it is bounded,
not because the scale-factor solution must end after one cycle.

The expansion time from the inner to the outer turning point is
\begin{equation}
\Delta\tau_{\rm exp}
=
g\int_{x_-}^{x_+}
\frac{dx}{\sqrt{\dfrac{2\mu x^2}{(1+x^2)^{3/2}}-\sin^2\psi_b}}.
\label{eq:closed_expansion_time_integral}
\end{equation}
For order-one matching data this time is of order the black-hole scale. It can
be made parametrically large by taking \(\psi_b\ll1\). In that limit
\(x_+\sim2\mu/\sin^2\psi_b\), and the large-\(x\) part of the integral gives
\begin{equation}
\Delta\tau_{\rm exp}
\sim
\frac{\pi M}{\sin^3\psi_b}
\qquad (\psi_b\ll1).
\label{eq:closed_expansion_time_scaling}
\end{equation}
This long expansion-time limit is a tuning of the matched branch rather than
a way to circumvent the result. The parameter $\psi_b$ is global junction
data fixing the comoving angular location of the FLRW boundary. Thus an
arbitrarily long expansion phase can be engineered by tuning the matching
construction, but it is not a prediction of the asymptotically flat parent
geometry itself.

The coalescing-root endpoint does not provide a second long-lived limit. Write
\[
\sin^2\psi_b=\frac{4\mu}{3\sqrt3}-\delta,
\qquad \delta\downarrow0.
\]
Expanding about $x_0=\sqrt2$ gives
\[
\frac{2\mu x^2}{(1+x^2)^{3/2}}-\sin^2\psi_b
=
\delta-\frac{4\sqrt3\,\mu}{27}(x-\sqrt2)^2
+O\!\left((x-\sqrt2)^3\right).
\]
It follows that
\[
\Delta\tau_{\rm exp}
\longrightarrow
\frac{3\,3^{1/4}\pi g}{2\sqrt\mu},
\]
which is finite. Moreover, for a nonextremal Bardeen black hole,
$\mu>3\sqrt3/4$ makes $4\mu/(3\sqrt3)>1$, so this equality is not reachable
with real matching data $\sin^2\psi_b\le1$. The only parametrically long
expansion regime exhibited here is therefore $\psi_b\to0$.

\subsection{Flat and open daughters avoid boundedness but fail completeness}
\label{subsec:flat_incomplete}

For the flat daughter branch, the large-\(A\) asymptotic from
Eqs.~\eqref{eq:flat_matching_equation} and \eqref{eq:asymptotic_flatness} is
\begin{equation}
H^2\sim \frac{2M}{A^3\chi_b^3},
\qquad
A(\tau)\propto \tau^{2/3}.
\label{eq:flat_late_H2}
\end{equation}
Thus the specific boundedness mechanism of the closed branch is absent. In
applying the flat/open completeness condition, the object being tested is the
maximal FLRW daughter spacetime determined by the matched scale-factor law
\(A(\tau)\), not merely the finite comoving ball bounded by the junction
surface. This maximal-FLRW test is not a completeness theorem for an otherwise
unspecified global parent-plus-daughter composite.

If the parent also has a de Sitter-like regular core,
\begin{equation}
F(R)=1-H_0^2R^2+O(R^4)
\qquad (R\to0),
\label{eq:desitter_core_parent}
\end{equation}
then Eq.~\eqref{eq:flat_matching_equation} implies
\begin{equation}
H^2\to H_0^2
\qquad (A\to0),
\label{eq:flat_early_desitter}
\end{equation}
so the early-time branch is asymptotically de Sitter-like,
\begin{equation}
A(\tau)\sim e^{H_0\tau}
\qquad (\tau\to-\infty).
\label{eq:flat_early_scale_factor}
\end{equation}
This makes the incompleteness mechanism explicit. For a radial null geodesic
in the flat FLRW metric, the conserved comoving momentum gives
\(d\lambda\propto A(\tau)d\tau\), where \(\lambda\) is an affine parameter.
Thus the past affine length of the asymptotic de Sitter branch is finite:
\begin{equation}
\Delta\lambda
\propto
\int_{-\infty}^{\tau_0} A(\tau)\,d\tau
\sim
\int_{-\infty}^{\tau_0} e^{H_0\tau}\,d\tau
<\infty .
\label{eq:flat_affine_incompleteness}
\end{equation}

Thus the flat daughter model is nonsingular in cosmic time and monotonic, but
it is past null incomplete via the general
flat/open FLRW completeness obstruction. More generally, non-static
curvature-regular flat/open FLRW spacetimes with regular affine ends cannot be
both null geodesically complete and ANEC-consistent
\cite{Burwig:2025hrr,Burwig:2026fsy}. Hence the flat daughter branch avoids
the closed-universe boundedness mechanism, but it does not satisfy our full
desired set of properties.

The open branch is in the same completeness class. As discussed in
Sec.~\ref{subsec:open_daughters}, it avoids the closed-branch boundedness
mechanism, but the local open de Sitter behavior is only a curvature-regular
FLRW patch and is not geodesically complete as an FLRW spacetime. Combining this
flat/open completeness obstruction with
Proposition~\ref{prop:closed_recollapse} proves
Theorem~\ref{thm:minimal_frw_daughters}: the closed branch is bounded by finite
ADM mass, while flat and open branches fail the regular-complete-ANEC conditions
by the general FLRW completeness theorem.

\subsection{Escape routes}
\label{subsec:escape_routes}

The assumptions enter the closed and flat/open models differently. In the
closed model, the large-$A$ failure follows from the finite-ADM limit
$m(R)\to M$, which makes the induced no-shell density scale as $A^{-3}$. An
unbounded closed daughter therefore requires an ingredient that changes this
late-time balance: modified asymptotics for which the effective mass function
does not approach a constant, an independent shell stress tensor, a non-FLRW or
non-comoving daughter geometry, or an additional smooth bulk component whose
density redshifts no faster than $A^{-2}$. A positive vacuum-energy component
is the simplest example, but adding it lies outside the asymptotically flat
minimal construction considered here.

Within the regular-end class of the cited theorem, a flat/open evasion must
give up at least one of the desired conditions: curvature regularity, null
completeness, ANEC consistency, the FLRW ansatz, or the flat/open curvature
class. A construction with irregular affine-end behavior instead lies outside
the theorem's stated hypotheses.

\section{Core type and source limitations}
\label{sec:discussion}

The preceding theorem is a structural statement about
asymptotically flat constructions, but does not apply to all
black-hole-universe models. Two possible escape attempts are especially
natural: changing the core type and using the parent matter source as the
daughter source.

\subsection{De Sitter cores versus Minkowski cores}
\label{subsec:core_types}

The Bardeen geometry has a locally de Sitter-like core, which is often one of
the motivations for viewing regular black-hole interiors as cosmological
seeds. This local behavior is compatible with recent convergence-condition
analyses: de Sitter-core regular black holes can preserve the null condition
while evading singularity theorems through strong-energy-condition violation
\cite{Borissova:2025hmj}. In our case, however, the core type
does not control the large-$A$ fate of the daughter.

Minkowski-core parents provide a useful comparison. Their early-time daughter
patch can approach a Minkowski-like nonsingular regime rather than an
asymptotically de Sitter-like one, as illustrated in
Ref.~\cite{LiWuGe:2025}. This changes the early-time patch but not the
closed-branch obstruction: for $k=+1$, boundedness follows from asymptotic
flatness and finite ADM mass. For $k=0$ and $k=-1$, changing the core type
likewise does not by itself evade the general flat/open completeness theorem.

\subsection{Can the parent matter source rescue the daughter universe?}
\label{subsec:source_rescue}

One might also ask whether the matter source associated with the parent
regular black hole can stabilize the daughter cosmology. We do not prove a
general no-go theorem, but the minimal Bardeen/nonlinear-electrodynamics route
is disfavored for two simple reasons.

First, the unchanged static source remains anisotropic away from the center and
decays too rapidly at large areal radius. For Bardeen, the static density falls
as $r^{-5}$, while the no-shell Friedmann density induced by the exterior mass
profile falls as $A^{-3}$. Neither behavior supplies the persistent
$A^{-2}$-or-slower support required to evade the closed-branch obstruction.

Second, promoting the same broad matter class to a homogeneous or coarse-grained
cosmological source is not automatic. A homogeneous magnetic sector must be
isotropized or averaged, and in an ordinary Maxwellian weak-field completion
the late-time behavior is radiation-like rather than dark-energy-like. The
original Ay\'on-Beato--Garc\'ia realization of Bardeen is not itself a standard
Maxwell weak-field model, so this should not be read as a no-go theorem for all
nonlinear-electrodynamics cosmologies. Our point is that neither the unchanged anisotropic parent source nor ordinary
weak-field homogeneous continuations naturally provide the persistent
$w\le -1/3$ support needed for an unbounded closed daughter.

\section{Conclusions}
\label{sec:conclusions}

We have examined whether asymptotically flat regular black holes can admit
minimal FLRW daughter cosmologies. We identify a structural obstruction. The
trapped interior of the parent spacetime is not itself an FLRW cosmology, so an
FLRW daughter must be introduced as a separate matched region. In the minimal
no-shell construction, closed daughters are bounded rather than indefinitely
expanding. Flat and open daughters can avoid this boundedness mechanism, but,
for their maximal homogeneous FLRW continuations within the regular endpoint
classes of the affine-ANEC theorem, non-static curvature-regular models cannot
be both null geodesically complete and ANEC-consistent.

These conclusions are not tied to one special metric. The closed-branch result
follows from asymptotic flatness and finite ADM mass, while the flat/open
obstruction follows from the stated completeness constraints on non-static FLRW
cosmologies. Changing the detailed core structure from de Sitter-like to
Minkowski-like modifies the early-time daughter patch but does not evade this
dichotomy within the same hypotheses. Likewise, the parent matter source does
not naturally provide the late-time support required to evade the
closed-branch obstruction. We have not proved that every separately specified
global spacetime formed from a finite FLRW cap and a parent region is
geodesically incomplete; that is a distinct global-extension problem.

Our result does not exclude more elaborate black-hole-universe constructions
with shells, modified asymptotics, non-FLRW daughters, or additional smooth bulk
components. Rather, it identifies which ingredients must be changed within the
minimal criteria adopted here. Closed FLRW daughters require additional
late-time support absent from finite-ADM no-shell matching, while flat/open FLRW
continuations must leave the hypotheses of the flat/open completeness theorem.
The main lesson is that a regular core, even a de Sitter-like one, is not by
itself sufficient to produce a viable, indefinitely expanding FLRW daughter
cosmology in the sense considered here.

%\newpage
\acknowledgments
DAE was supported in part by the U.S. Department of Energy,
Office of High Energy Physics, under Award Number DE-SC0019470.

%\newpage
\appendix

\section{Darmois conditions for the comoving no-shell boundary}
\label{app:darmois}

Here we present the full no-shell matching conditions used in
Sec.~\ref{subsec:general_junction}. For the FLRW daughter metric
\eqref{eq:generic_frw_metric}, a boundary at fixed \(\chi=\chi_b\) has areal
radius \(R_b=A\chi_b\). The outward unit normal on the FLRW side is
\begin{equation}
n_-^\mu\partial_\mu=\frac{\sqrt{1-k\chi_b^2}}{A}\,\partial_\chi,
\end{equation}
and therefore
\begin{equation}
K_{\theta\theta}^- = R_b\,n_-^\mu\partial_\mu R
=R_b\sqrt{1-k\chi_b^2}.
\label{eq:appendix_Ktheta_minus}
\end{equation}
On the parent side, with trajectory
\(x_+^\mu=(T(\tau),R_b(\tau),\theta,\varphi)\), the induced-metric condition is
\begin{equation}
F(R_b)\dot T^{\,2}-\frac{\dot R_b^{\,2}}{F(R_b)}=1.
\end{equation}
Choosing the outward normal consistently with the orientation in the main text
gives
\begin{equation}
K_{\theta\theta}^+=R_b F(R_b)\dot T
=R_b\sqrt{F(R_b)+\dot R_b^{\,2}}.
\label{eq:appendix_Ktheta_plus}
\end{equation}
Equating Eqs.~\eqref{eq:appendix_Ktheta_minus} and
\eqref{eq:appendix_Ktheta_plus} yields
\begin{equation}
\dot R_b^{\,2}+F(R_b)=1-k\chi_b^2,
\end{equation}
which is Eq.~\eqref{eq:rr_matching_condition}. For the closed parametrization
\(R_b=A\sin\psi_b\), the same calculation gives
\begin{equation}
K_{\theta\theta}^- = R_b\cos\psi_b,
\qquad
\dot R_b^{\,2}+F(R_b)=\cos^2\psi_b .
\end{equation}

For a smooth nonstatic trajectory, the remaining \(\tau\tau\) component is
consistent with the same comoving no-shell construction; the identically
static case is treated separately in this appendix. On the FLRW side the boundary is comoving and geodesic, hence
\(K_{\tau\tau}^-=0\). On the parent side, differentiating the first integral
away from an isolated turning point gives
\begin{equation}
\ddot R_b=-\frac12 F'(R_b),
\label{eq:appendix_boundary_geodesic_equation}
\end{equation}
which is precisely the radial geodesic equation for the trajectory in the
one-function parent metric. Thus \(K_{\tau\tau}^+=0\). At an isolated turning
point this relation is understood by continuity, or equivalently by using the
local second-order geodesic equation rather than dividing by \(\dot R_b\).
Hence the angular condition together with induced-metric matching gives the
full Darmois no-shell matching for a smooth nonstatic comoving construction.

An identically static boundary is exceptional: differentiation of the first
integral gives no information. For a timelike static boundary at $R_b=R_0$,
full matching additionally requires $F'(R_0)=0$ in the one-function metric,
besides the angular condition $F(R_0)=1-k\chi_b^2$ (or
$F(R_0)=\cos^2\psi_b$ in the closed parametrization). For the redshifted
metric \eqref{eq:general_redshift_metric}, the angular condition is unchanged,
while the \(\tau\tau\) condition gives the additional constraint discussed in
Sec.~\ref{subsec:redshift_functions}.


\begin{thebibliography}{99}

%\cite{Wall:2009wi}
\bibitem{Wall:2009wi}
A.~C.~Wall,
``Proving the Achronal Averaged Null Energy Condition from the Generalized Second Law,''
Phys. Rev. D \textbf{81}, 024038 (2010)
doi:10.1103/PhysRevD.81.024038
[arXiv:0910.5751 [gr-qc]].
%78 citations counted in INSPIRE as of 19 Jun 2026

\bibitem{Curiel:2014zba}
E.~Curiel,
``A Primer on Energy Conditions,''
Einstein Stud. \textbf{13}, 43--104 (2017)
doi:10.1007/978-1-4939-3210-8{\_}3
[arXiv:1405.0403 [physics.hist-ph]].

\bibitem{Burwig:2025hrr}
N.~L.~Burwig and D.~A.~Easson,
``Open case for a closed universe,''
Phys. Rev. D \textbf{113}, no.8, 083530 (2026)
doi:10.1103/mn3v-myzc
[arXiv:2510.13971 [hep-th]].

%\cite{Burwig:2026fsy}
\bibitem{Burwig:2026fsy}
N.~L.~Burwig and D.~A.~Easson,
``Affine ANEC selects the closed FRW branch for geodesically complete cosmology,''
[arXiv:2605.18965 [gr-qc]].
%1 citations counted in INSPIRE as of 12 Jun 2026

\bibitem{Pathria:1972}
R.~K.~Pathria,
``The Universe as a Black Hole,''
Nature {\bf 240}, 298 (1972).

%\cite{Easson:2001qf}
\bibitem{Easson:2001qf}
D.~A.~Easson and R.~H.~Brandenberger,
``Universe generation from black hole interiors,''
JHEP \textbf{06}, 024 (2001)
doi:10.1088/1126-6708/2001/06/024
[arXiv:hep-th/0103019 [hep-th]].
%101 citations counted in INSPIRE as of 19 Jun 2026

%\cite{Blau:1986cw}
\bibitem{Blau:1986cw}
S.~K.~Blau, E.~I.~Guendelman and A.~H.~Guth,
``The Dynamics of False Vacuum Bubbles,''
Phys. Rev. D \textbf{35}, 1747 (1987)
doi:10.1103/PhysRevD.35.1747
%457 citations counted in INSPIRE as of 12 Jun 2026

\bibitem{Frolov:1989}
V.~P.~Frolov, M.~A.~Markov, and V.~F.~Mukhanov,
``Through a black hole into a new universe?''
Phys.\ Lett.\ B {\bf 216}, 272 (1989).

\bibitem{Frolov:1990}
V.~P.~Frolov, M.~A.~Markov, and V.~F.~Mukhanov,
``Black holes as possible sources of closed and semiclosed worlds,''
Phys.\ Rev.\ D {\bf 41}, 383 (1990).

\bibitem{FarhiGuthGuven:1990}
E.~Farhi, A.~H.~Guth, and J.~Guven,
``Is it possible to create a universe in the laboratory by quantum tunneling?''
Nucl.\ Phys.\ B {\bf 339}, 417 (1990).

\bibitem{Trodden:1993xh}
M.~Trodden, V.~F.~Mukhanov and R.~H.~Brandenberger,
``A nonsingular two dimensional black hole,''
Phys. Lett. B \textbf{316}, 483 (1993)
[arXiv:hep-th/9305111].

%\cite{Easson:2017pfe}
\bibitem{Easson:2017pfe}
D.~A.~Easson,
``Nonsingular Schwarzschild{\textendash}de Sitter black hole,''
Class. Quant. Grav. \textbf{35}, no.23, 235005 (2018)
doi:10.1088/1361-6382/aae85f
[arXiv:1712.09455 [hep-th]].
%17 citations counted in INSPIRE as of 12 Jun 2026

%\cite{Frolov:2021kcv}
\bibitem{Frolov:2021kcv}
V.~P.~Frolov and A.~Zelnikov,
``Two-dimensional black holes in the limiting curvature theory of gravity,''
JHEP \textbf{08}, 154 (2021)
doi:10.1007/JHEP08(2021)154
[arXiv:2105.12808 [hep-th]].
%18 citations counted in INSPIRE as of 19 Jun 2026


\bibitem{Dymnikova:2001}
I.~G.~Dymnikova, A.~Dobosz, M.~L.~Fil'chenkov, and A.~Gromov,
``Universes inside a \(\Lambda\) black hole,''
Phys.\ Lett.\ B {\bf 506}, 351 (2001)
[arXiv:gr-qc/0102032].

\bibitem{Dymnikova:2019}
I.~Dymnikova,
``Universes Inside a Black Hole with the de Sitter Interior,''
Universe {\bf 5}, 111 (2019).

\bibitem{Bardeen:1968}
J.~M.~Bardeen,
``Non-singular general-relativistic gravitational collapse,''
in {\it Proceedings of the International Conference GR5},
Tbilisi, U.S.S.R. (1968).

\bibitem{AyonBeatoGarcia:2000}
E.~Ay{\'o}n-Beato and A.~Garc{\'\i}a,
``The Bardeen model as a nonlinear magnetic monopole,''
Phys.\ Lett.\ B {\bf 493}, 149 (2000)
[arXiv:gr-qc/0009077].

\bibitem{Hayward:2006}
S.~A.~Hayward,
``Formation and evaporation of regular black holes,''
Phys.\ Rev.\ Lett.\ {\bf 96}, 031103 (2006)
[arXiv:gr-qc/0506126].

\bibitem{Ansoldi:2008}
S.~Ansoldi,
``Spherical black holes with regular center: A review of existing models including a recent realization with Gaussian sources,''
[arXiv:0802.0330 [gr-qc]].

\bibitem{BalartVagenas:2014}
L.~Balart and E.~C.~Vagenas,
``Regular black hole metrics and the weak energy condition,''
Phys.\ Lett.\ B {\bf 730}, 14 (2014)
[arXiv:1401.2136 [gr-qc]].

\bibitem{Fan:2016rih}
Z.-Y.~Fan and X.~Wang,
``Construction of regular black holes in general relativity,''
Phys. Rev. D \textbf{94}, no.12, 124027 (2016)
doi:10.1103/PhysRevD.94.124027
[arXiv:1610.02636 [gr-qc]].

\bibitem{Borde:1994}
A.~Borde and A.~Vilenkin,
``Eternal inflation and the initial singularity,''
Phys. Rev. Lett. \textbf{72}, 3305 (1994)
doi:10.1103/PhysRevLett.72.3305
[arXiv:gr-qc/9312022].

\bibitem{BordeGuthVilenkin:2003}
A.~Borde, A.~H.~Guth and A.~Vilenkin,
``Inflationary spacetimes are incomplete in past directions,''
Phys. Rev. Lett. \textbf{90}, 151301 (2003)
doi:10.1103/PhysRevLett.90.151301
[arXiv:gr-qc/0110012].

\bibitem{LesnefskyEassonDavies:2023}
J.~E.~Lesnefsky, D.~A.~Easson, and P.~C.~W.~Davies,
``On the past-completeness of inflationary spacetimes,''
Phys.\ Rev.\ D {\bf 107}, 044024 (2023)
[arXiv:2207.00955 [gr-qc]].

\bibitem{Easson:2024uxe}
D.~A.~Easson and J.~E.~Lesnefsky,
``Inflationary resolution of the initial singularity,''
Phys. Lett. B \textbf{875}, 140370 (2026)
doi:10.1016/j.physletb.2026.140370
[arXiv:2402.13031 [hep-th]].

\bibitem{Easson:2024fzn}
D.~A.~Easson and J.~E.~Lesnefsky,
``Eternal universes,''
Phys. Rev. D \textbf{112}, no.6, 063545 (2025)
doi:10.1103/5mhz-m8bg
[arXiv:2404.03016 [hep-th]].

%\cite{Easson:2026ret}
\bibitem{Easson:2026ret}
D.~A.~Easson,
``Geodesically Complete Curvature-Bounce Inflation,''
[arXiv:2604.27103 [astro-ph.CO]].
%1 citations counted in INSPIRE as of 12 Jun 2026

\bibitem{OppenheimerSnyder:1939}
J.~R.~Oppenheimer and H.~Snyder,
``On Continued Gravitational Contraction,''
Phys.\ Rev.\ {\bf 56}, 455 (1939).

\bibitem{Israel:1966}
W.~Israel,
``Singular hypersurfaces and thin shells in general relativity,''
Nuovo Cim.\ B {\bf 44}, 1 (1966).

\bibitem{MisnerSharp:1964}
C.~W.~Misner and D.~H.~Sharp,
``Relativistic equations for adiabatic, spherically symmetric gravitational collapse,''
Phys.\ Rev.\ {\bf 136}, B571 (1964).

\bibitem{Shojai:2022}
F.~Shojai, A.~Sadeghi, and R.~Hassannejad,
``Generalized Oppenheimer-Snyder Gravitational Collapse into Regular Black holes,''
Class.\ Quantum Grav.\ {\bf 39}, 085003 (2022)
[arXiv:2202.14024 [gr-qc]].

\bibitem{Bobula:2024}
M.~Bobula,
``Cosmic inflation prevents singularity formation in collapse into a Hayward black hole,''
Class.\ Quantum Grav.\ {\bf 42}, 065027 (2025)
[arXiv:2404.12243 [gr-qc]].

\bibitem{Bueno:2025}
P.~Bueno, P.~A.~Cano, R.~A.~Hennigar, {\'A}.~J.~Murcia and A.~Vicente-Cano,
``Regular black holes from Oppenheimer-Snyder collapse,''
Phys. Rev. D \textbf{112}, no.6, 064039 (2025)
doi:10.1103/qrbb-mdvm
[arXiv:2505.09680 [gr-qc]].

\bibitem{LiWuGe:2025}
S.~Li, J.-P.~Wu, and X.-H.~Ge,
``Non-singular cosmologies matching regular black holes,''
[arXiv:2512.00926 [gr-qc]].

%\cite{Bobula:2026zlq}
\bibitem{Bobula:2026zlq}
M.~Bobula and F.~Fazzini,
``Quantum gravitational stellar evolution beyond shell-crossing singularities,''
Phys. Rev. D \textbf{113}, no.10, 106016 (2026)
doi:10.1103/165q-lwmd
[arXiv:2601.18618 [gr-qc]].
%2 citations counted in INSPIRE as of 21 Jun 2026

\bibitem{Borissova:2025hmj}
J.~Borissova, S.~Liberati and M.~Visser,
``Timelike convergence condition in regular black-hole spacetimes with (anti-)de Sitter core,''
Phys. Rev. D \textbf{112}, no.10, 104072 (2025)
doi:10.1103/rrc9-g1sv
[arXiv:2509.08590 [gr-qc]].

%\cite{Easson:2020esh}
\bibitem{Easson:2020esh}
D.~A.~Easson, C.~Keeler and T.~Manton,
``Classical double copy of nonsingular black holes,''
Phys. Rev. D \textbf{102}, no.8, 086015 (2020)
doi:10.1103/PhysRevD.102.086015
[arXiv:2007.16186 [gr-qc]].
%73 citations counted in INSPIRE as of 12 Jun 2026

\bibitem{Easson:2026wex}
D.~A.~Easson and T.~Manton,
``Black Hole Interiors as a Laboratory for Time-Dependent Classical Double Copy,''
[arXiv:2604.19920 [hep-th]].



\end{thebibliography}
\end{document}